\documentclass[12pt]{article}
\usepackage[left=35mm,right=35mm,top=35mm,bottom=35mm]{geometry}
\usepackage{amsmath}
\usepackage{amsthm}
\usepackage{amssymb}
\usepackage{amsfonts}
\usepackage{mathrsfs}

\DeclareMathOperator*{\slim}{s-lim}
\DeclareMathOperator{\supp}{supp}

\makeatletter

\@addtoreset{equation}{section}
\makeatother

\newtheorem{ass}{Assumption}[section]
\newtheorem{thm}[ass]{Theorem}
\newtheorem{prop}[ass]{Proposition}
\newtheorem{lem}[ass]{Lemma}

\begin{document}
\begin{flushleft}
{\Large \bf Inverse scattering for $N$-body time-decaying harmonic oscillators}
\end{flushleft}

\begin{flushleft}
{\large Atsuhide ISHIDA}\\
{Katsushika Division, Institute of Arts and Sciences, Tokyo University of Science, 6-3-1 Niijuku, Katsushika-ku, Tokyo 125-8585, Japan\\ 
Email: aishida@rs.tus.ac.jp
}
\end{flushleft}

\begin{abstract}
In the previous study \cite{Is7}, the author proved the uniqueness of short-range potential functions using the Enss-Weder time-dependent method \cite{EnWe} for a two-body quantum system described by time-decaying harmonic oscillators. In this study, we extend the result of \cite{Is7} to the $N$-body case. We use the approaches developed in \cite{EnWe, VaWe, We} to prove that the high-velocity limit of the scattering operator uniquely determines all the pairwise interaction potentials among the $N$ particles, focusing respectively on each fixed pair of particles.
 \end{abstract}

\quad\textit{Keywords}: $N$-body scattering theory, wave operator, scattering operator\par
\quad\textit{MSC}2020: 35R30, 81U10, 81U40

\section{Introduction\label{introduction}}
We consider a quantum system governed by the $N$-particle harmonic oscillators and interactional potential functions described by the Hamiltonian
\begin{equation}
\tilde{H}(t)=-\sum_{j=1}^N\Delta_{r_j}/(2m_j)+k(t)\sum_{j=1}^Nm_jr_j^2/2+\sum_{1\leqslant j<k\leqslant N}V_{jk}(r_k-r_j)\label{perturbed1}
\end{equation}
acting on $L^2({\mathbb{R}^{dN}})$ with $N\geqslant2$ and $d\geqslant2$, where $r_j\in\mathbb{R}^d$ is the position of the $j$th particle, $m_j>0$ is the mass, $\Delta_{r_j}$ is the Laplacian for $r_j$, and $V_{jk}$ is the interaction between $j$th and $k$th particles. The time-dependent coefficient $k(t)$ is defined as
\begin{equation}
k(t)=
\begin{cases}
\ \omega^2 & \quad \mbox{if}\quad|t|< T,\\
\ \sigma/t^2 & \quad \mbox{if}\quad|t|\geqslant T\label{coeff}
\end{cases}
\end{equation}
for $0<\sigma<1/4$, $\omega>0$ and $T>0$.

We now introduce the Jacobi coordinates in which the center of mass of all $N$ particles is located at the origin (see \cite[Chapter 3]{Iso} or \cite[XI.5]{ReSi2}). Let $\mathscr{X}$ be the subspace of $\mathbb{R}^{dN}$ described by
\begin{equation}
\mathscr{X}=\left\{(r_1,\ldots,r_N)\in\mathbb{R}^{dN}\Biggm|\sum_{j=1}^Nm_jr_j=0\right\}\simeq\mathbb{R}^{d(N-1)}
\end{equation}
 and $\mathscr{X}_{\rm cm}$ be the orthogonal space for $\mathscr{X}$ associated with the scalar product
\begin{equation}
(r,\tilde{r})_{\mathbb{R}^{dN}}=\sum_{j=1}^Nm_jr_j\cdot\tilde{r}_j\label{scalar_product}
\end{equation}
on $\mathbb{R}^{dN}$ for $r=(r_1,\ldots,r_N)$ and $\tilde{r}=(\tilde{r}_1,\ldots,\tilde{r}_N)\in\mathbb{R}^{dN}$. When the position of the center of mass is removed, the unperturbed Hamiltonian of \eqref{perturbed1} is given by
\begin{equation}
H_0(t)=-\sum_{j=1}^N\Delta_{r_j}/(2m_j)+\Delta_{\rm cm}/(2M)+k(t)\sum_{j=1}^Nm_jr_j^2/2-k(t)Mr_{\rm cm}^2/2\label{free1}
\end{equation}
acting on $L^2(\mathscr{X})$, where $M=\sum_{j=1}^nm_j$ is the total mass, $r_{\rm cm}=\sum_{j=1}^Nm_jr_j/M\in\mathscr{X}_{\rm cm}$ is the position of the center of mass, and $\Delta_{\rm cm}$ is the Laplace-Beltrami operator on $\mathscr{X}_{\rm cm}$. This expression \eqref{free1} can be simplified by rewriting it in Jacobi coordinates as
\begin{equation}
y_j=r_{j+1}-\sum_{k=1}^jm_kr_k/\sum_{k=1}^jm_k\in\mathscr{X}\label{jacobi1}
\end{equation}
for $1\leqslant j\leqslant N-1$. The reduced mass $\mu_j$ is defined as
\begin{equation}
1/\mu_j=1/\sum_{k=1}^jm_k+1/m_{j+1}.
\end{equation}
Under these notations, we have
\begin{gather}
\sum_{j=1}^Nm_jr_j^2=Mr_{\rm cm}^2+\sum_{j=1}^{N-1}\mu_jy_j^2,\\
\sum_{j=1}^N\Delta_{r_j}/m_j=\Delta_{\rm cm}/M+\sum_{j=1}^{N-1}\Delta_{y_j}/\mu_j
\end{gather}
where $\Delta_{y_j}$ is the Laplace-Beltrami operator on $\mathscr{X}$. Therefore, in these coordinates, \eqref{free1} can be written as the unperturbed Hamiltonian
\begin{equation}
H_0(t)=-\sum_{j=1}^{N-1}\Delta_{x_j}/2+k(t)\sum_{j=1}^{N-1}x_j^2/2\label{free2}
\end{equation}
acting on $L^2(\mathscr{X})$ with $x_j=\sqrt{\mu_j}y_j$ for simplicity.

We now give the assumptions for each potential function $V_{jk}$ in \eqref{perturbed1} as a perturbation of $H_0(t)$. Let $\lambda$ satisfy
\begin{equation}
0<\lambda=(1-\sqrt{1-4\sigma})/2<1/2.
\end{equation}
This $\lambda$ describes the effective scaling rate induced by the time-decaying harmonic potential and is directly related to the classical trajectories $x(t)=O(|t|^{1-\lambda})$ as $|t|\rightarrow\infty$. The condition $0<\lambda<1/2$ ensures the short-range character of the interaction under the scaling associated with the propagator. The bracket $\langle\cdot\rangle$ is defined by $\sqrt{1+|\cdot|}$ and $A\lesssim B$ means that there exists a constant $C>0$ such that $A\leqslant CB$.

\begin{ass}\label{ass}
Each function $V_{jk}:\mathbb{R}^d\rightarrow\mathbb{R}$ for $1\leqslant j<k\leqslant N$ is a multiplication operator and decomposed into two parts:
\begin{equation}
V_{jk}=V_{jk}^{\rm bdd}+V_{jk}^{\rm sing}.
\end{equation}
The bounded part $V_{jk}^{\rm bdd}\in L^\infty(\mathbb{R}^d)$ satisfies
\begin{equation}
|V_{jk}^{\rm bdd}(x)|\lesssim\langle x\rangle^{-\rho}
\end{equation}
with $\rho>1/(1-\lambda)$. The singular part $V_{jk}^{\rm sing}\in L^q(\mathbb{R}^d)$ is compactly supported on $\mathscr{X}$, where the Lebesgue exponent $q$ satisfies
\begin{equation}
\infty>q
\begin{cases}
\ =2\quad & \mbox{\rm if}\quad n\leqslant3,\\
\ >n/2\quad & \mbox{\rm if}\quad n\geqslant4.\label{sing}
\end{cases}
\end{equation}
\end{ass}

It is well-known that each $V_{jk}^{\rm sing}$ is $-\Delta$-bounded infinitesimally. The interacting Hamiltonian is
\begin{equation}
H(t)=H_0(t)+\sum_{1\leqslant j<k\leqslant N}V_{jk}(r_k-r_j)\label{perturbed2}
\end{equation}
acting on $L^2(\mathscr{X})$. Here we note $r_k-r_j\in\mathscr{X}$ for any $1\leqslant j<k\leqslant N$. Under these definitions, the perturbed full Hamiltonian \eqref{perturbed1} is decomposed into
\begin{equation}
\tilde{H}(t)=H(t)\otimes1+1\otimes(-\Delta_{\rm cm}/(2M)-k(t)Mx_{\rm cm}^2/2)
\end{equation}
on $L^2(\mathscr{X})\otimes L^2({\mathscr{X}_{\rm cm}})$, where $1$ denotes the identities on $L^2(\mathscr{X})$ and $L^2(\mathscr{X}_{\rm cm})$. Therefore, we will concentrate on studying the quantum system governed by $H(t)$. Because two parameter unitary propagators generated by $H_0(t)$ and $H(t)$ exist uniquely by the result of \cite[Theorem 6 and Remark (a)]{Ya}, we denote these propagators as $U_0(t,s)$ and $U(t,s)$. In $N$-body quantum scattering theory, we have to consider the cluster decomposition $a=\{C_1,\ldots,C_m\}$ with $m\geqslant2$, where $C_j\subset\{1,\ldots,N\}$ for $1\leqslant j\leqslant m$, which satisfies $\cup_{j=1}^mC_j=\{1,\ldots,N\}$ and $C_j\cap C_k=\emptyset$ if $j\not=k$. Let $\mathscr{X}^a$ be the subspace of $\mathscr{X}$ given by
\begin{equation}
\mathscr{X}^a=\left\{(r_1,\ldots,r_N)\in\mathscr{X}\Biggm|\sum_{j\in C}m_jr_j=0,\ C\in a\right\}
\end{equation}
and $\mathscr{X}_a$ be its orthogonal space. The cluster Hamiltonian is defined as
\begin{equation}
H_a(t)=-\sum_{j=1}^{N-1}\Delta_{x_j}/2+k(t)\sum_{j=1}^{N-1}x_j^2/2+\sum_{\{j,k\}\subset a}V_{jk}(r_k-r_j)\label{cluster_hamiltonian}
\end{equation}
acting on $L^2(\mathscr{X})$. In the subspace $\mathscr{X}_a$, $V_{jk}$ does not vanish for $\{j,k\}\subset a$. Therefore, it is natural that the cluster interaction
\begin{equation}
\sum_{\{j,k\}\not\subset a}V_{jk}(r_k-r_j)
\end{equation}
be regarded as the perturbation and that the cluster Hamiltonian $H_a(t)$ of \eqref{cluster_hamiltonian} be considered the unperturbed system for any cluster decomposition $a$. However, to discuss inverse scattering, it suffices to consider the $N$-cluster case where $a=\{\{1\},\ldots,\{N\}\}$ and the wave operators defined as
\begin{equation}
W^\pm=\slim_{t\rightarrow\infty}U(t,0)^*U_0(t,0)\label{wave_operator1}.
\end{equation}
We give the proof of existence of \eqref{wave_operator1} in Section \ref{existence}. Then the scattering operator is defined as
\begin{equation}
S(V)=(W^+)^*W^-\label{scattering_operator}
\end{equation}
for $V=\sum_{1\leqslant j<k\leqslant N}V_{jk}(r_k-r_j)$.

\begin{thm}\label{thm1}
Let $V_1=\sum_{1\leqslant j<k\leqslant N}V_{jk,1}$ and $V_2=\sum_{1\leqslant j<k\leqslant N}V_{jk,2}$ satisfy Assumption \ref{ass}. If $S(V_1)=S(V_2)$, then $V_{jk,1}=V_{jk,2}$ for $1\leqslant j<k\leqslant N$.
\end{thm}

Many researchers have studied scattering theory under the $N$-body quantum systems, in particular, the the existence and asymptotic completeness of the wave operators, has been studied. Regarding the standard $N$-body Schr\"odinger operators, see \cite{AdItItSk, De1, De2, Gr, SiSo1, SiSo2, Sk}; the detailed history and proofs are also provided in \cite{DeGe, Iso}. For $N$-body Schr\"odinger operators in the external electric fields, see \cite{Ad1, Ad3, AdTa1, AdTa2, HeMoSk}. For the external magnetic fields, see \cite{Ad2, GeLa}. However, there are no results regarding scattering theory for $N$-body time-decaying harmonic oscillators.

The Enss-Weder time-dependent method was invented in \cite{EnWe} for the standard Sch\"odinger operators. Since then, many researchers have applied this method to other quantum models. The Stark effect and time-dependent electric fields were studied by \cite{AdFuIs, AdKaKaTo, AdMa, AdTs1, AdTs2, Is3, Is6, Ni1, Ni3, VaWe}. The repulsive Hamiltonians were studied by \cite{Is1, Is5, Ni3}. The fractional Laplacian and Dirac equations were studied by \cite{Is4, Ju}. The non-linear Schr\"odinger equations were studied by \cite{Wa1}. The Hartree-Fock equations were studied by \cite{Wa2, Wa3}. Recently, the author investigated the two-body time-decaying harmonic oscillator and repulsive Hamiltonian in \cite{Is7, Is8}. $N$-body inverse scattering was investigated by \cite{EnWe} for the standard Schr\"odinger case, and by \cite{We, VaWe} for the Stark effect. Our proof also follows the strategies of these previous studies. That is, by an appropriate choice of Jacobi coordinates, the problem is reduced to the two-body case \cite{Is7}. It should be emphasized that the pioneering works of \cite{EnWe, We} already treated both the two-body and $N$-body case in a unified manner. The arguments of \cite{EnWe, We, VaWe} show that the reduction to two-body channels is a fundamental principle.

Regarding scattering theory for the two-body time-decaying harmonic oscillator, \cite{IsKa1} proved the existence of the wave operators under the assumption \ref{ass} only for the bounded part $V^{\rm bdd}$ and non-existence for $\rho\leqslant1/(1-\lambda)$, and clarified that the condition $\rho>1/(1-\lambda)$ is short-range. We can intuitively understand the threshold $1/(1-\lambda)$ from the classical motion of the particle $x(t)=c_1t^{1-\lambda}+c_2t^\lambda$ for $t\geqslant T$, which satisfies the Newton equation $({\rm d}^2/{\rm d}t^2)x(t)=-k(t)x(t)$ because $\lambda$ is one of the roots of the quadratic equation $\lambda^2-\lambda+\sigma=0$. The critical case $\sigma=1/4$ was studied by \cite{IsKa2}. For the inverse square potential, \cite{IsKa3} proved the asymptotic completeness of the wave operators. Moreover, \cite{IsKa3} proved the Strichartz estimates and applied them to the initial-value problems for the non-linear Schr\"odinger equations. The non-linear Schr\"odinger equations and Strichartz estimates with time-decaying harmonic oscillators have also been studied by \cite{Ka1, Ka2, Ka3, KaMi1, KaMi2, KaMu, KaSa, KaYo}.

In contrast to the time-independent harmonic case or standard free Schr\"odinger dynamics, the time-decaying quadratic potential induces a nontrivial scaling behavior in the classical trajectories and the quantum propagator. One of the key technical ingredients of this paper is a propagation estimate for the free dynamics associated with the time-dependent harmonic oscillator in the $N$-body framework (Proposition \ref{prop3}). Although the estimate itself can be reduced to the two-body case by an appropriate choice of Jacobi coordinates, its formulation in a form suitable for the $N$-body setting is essential for proving the existence of the $N$-cluster wave operators. Once the existence of the wave operators is established, the inverse scattering argument follows the Enss-Weder time-dependent method and reduces the problem to suitable two-body channels. While the inverse problem itself is treated by adapting known high-velocity arguments, the derivation of the propagation estimates (Lemmas \ref{lem7} and \ref{lem8}) in the presence of a time-decaying harmonic potential and their extension from the two-body case studied in \cite{Is7} to the $N$-body setting constitute the main technical novelty of the present work.

Throughout this paper, we identify $L^2(\mathscr{X})$ with $L^2(\mathbb{R}^{d(N-1)})$ and use the following notations. $\|\cdot\|$ denotes the $L^2$-norm or operator norm on $L^2(\mathbb{R}^{d(N-1)})$, and $(\cdot,\cdot)$ denotes the scalar product of $L^2(\mathbb{R}^{d(N-1)})$. $F(\cdots)$ is the characteristic function of the set $\{\cdots\}$.

\section{Existence of Wave Operators}\label{existence}
In this section, we prove the existence of the $N$-cluster wave operators \eqref{wave_operator1}. To do that, we reduce the problem to a much simpler form of strong limits. We define $D=(D_1,\ldots,D_{N-1})$ for $D_j=-{\rm i}\nabla_{x_j}$ and $x=(x_1,\ldots,x_{N-1})\in\mathbb{R}^{d(N-1)}$. Then the free Hamiltonian is written as
\begin{equation}
H_0(t)=D^2/2+k(t)x^2/2=\sum_{j=1}^{N-1}D_j^2/2+k(t)\sum_{j=1}^{N-1}x_j^2/2.\label{free3}
\end{equation}
The propagators $U_0(t,s)$ and $U(t,s)$ have unitary factorizations for $t,s\geqslant T$ or $t,s\leqslant-T$ that were proved by \cite[Proposition 1]{IsKa1} in the two-body case. We define
\begin{equation}
U_{0,\lambda}(t)=e^{{\rm i}\lambda x^2/(2t)}e^{-{\rm i}\lambda\log tA}e^{-{\rm i}t^{1-2\lambda}D^2/(2(1-2\lambda))}\label{factorization+}
\end{equation}
if $t\geqslant T$ and
\begin{equation}
U_{0,\lambda}(t)=e^{{\rm i}\lambda x^2/(2t)}e^{-{\rm i}\lambda\log(-t)A}e^{{\rm i}(-t)^{1-2\lambda}D^2/(2(1-2\lambda))}\label{factorization-}
\end{equation}
if $t\leqslant-T$, where $A=(D\cdot x+x\cdot D)/2$. Then the factorizations
\begin{equation}
U_0(t,s)=U_{0,\lambda}(t)U_{0,\lambda}(s)^*\label{factorization_free}
\end{equation}
and 
\begin{equation}
U(t,s)=e^{{\rm i}\lambda x^2/(2t)}e^{-{\rm i}\lambda\log|t|A}U_\lambda(t,s)e^{{\rm i}\lambda\log|s|A}e^{-{\rm i}\lambda x^2/(2s)}\label{factorization_full}
\end{equation}
hold for $t,s\geqslant T$ or $t,s\leqslant-T$, where $U_\lambda(t,s)$ is the propagator generated by
\begin{equation}
D^2/(2|t|^{2\lambda})+\sum_{1\leqslant j<k \leqslant N}V_{jk}(|t|^\lambda(r_k-r_j)).
\end{equation}
Regarding the proofs of \eqref{factorization_free} and \eqref{factorization_full}, it suffices to prove the following Proposition \ref{prop2} because $e^{{\rm i}\lambda\log|t|A}D_je^{-{\rm i}\lambda\log|t|A}=D_j/|t|^\lambda$ for $1\leqslant j\leqslant N$ clearly holds, and the rest of the proofs are demonstrated in \cite[Proposition 1]{IsKa1}.

\begin{prop}\label{prop2}
\begin{equation}
e^{{\rm i}\lambda\log|t|A}(r_k-r_j)e^{-{\rm i}\lambda\log|t|A}=|t|^\lambda(r_k-r_j)\label{prop2_1}
\end{equation}
holds for any $1\leqslant j<k\leqslant N$.
\end{prop}

\begin{proof}[Proof of Proposition \ref{prop2}]
For $1\leqslant j<k\leqslant N$ fixed, we construct the Jacobi coordinates from $y_1=r_k-r_j$ and define $A_1=(x_1\cdot D_1+D_1\cdot x_1)/2$ for $x_1=\sqrt{\mu_1}y_1$. We have
\begin{equation}
e^{{\rm i}\lambda\log|t|A}y_1e^{-{\rm i}\lambda\log|t|A}=e^{{\rm i}\lambda\log|t|A_1}x_1e^{-{\rm i}\lambda\log|t|A_1}/\sqrt{\mu_1}.
\end{equation}
As in the proof of \cite[Proposition 1]{IsKa1}, we have \eqref{prop2_1}.
\end{proof}

To prove the existence of the wave operators \eqref{wave_operator1}, it suffices to prove the existence of the strong limits
\begin{equation}
\slim_{t\rightarrow\pm\infty}U_{\lambda}(t,T)^*e^{-{\rm i}t^{1-2\lambda}D^2/(2(1-2\lambda))}\label{wave_operator2}
\end{equation}
by virtue of the factorizations \eqref{factorization_free} and \eqref{factorization_full}, and the chain rules of the propagators. The following propagation estimate and the existence of the wave operators for the two-body time-decaying harmonic oscillator with only the bounded part $V^{\rm bdd}$ only were proved in \cite[Proposition 1 and Theorem 1]{IsKa1} by establishing the correspondence with \eqref{wave_operator2} for the two-body case. Our proof is one of the extensions to the $N$-body case with the singular part $V_{jk}^{\rm sing}$.

Here we introduce a conical region in which all relative momenta are uniformly bounded away from zero. We denote a $d(N-1)$-dimensional open ball with radius $r>0$ centered at $\xi\in\mathbb{R}^{d(N-1)}$ as $B_{d(N-1)}(\xi,r)$ and its spherical surface as $\mathbb{S}^{d(N-1)-1}$. For any $1\leqslant j<k\leqslant N$, there exists a full-rank matrix $L_{jk}:\mathbb{R}^{d(N-1)}\rightarrow\mathbb{R}^d$ such that any conjugate momentum $\xi_{jk}\in\mathbb{R}^d$ associated with the relative coordinate $r_k-r_j$ is given by $\xi_{jk}=L_{jk}\xi$ for some $\xi\in\mathbb{R}^{d(N-1)}$. We note that $L_{jk}$ is surjective and that $\dim\ker L_{jk}=d(N-2)$. Let $\omega_0\in\mathbb{S}^{d(N-1)-1}\setminus(\cup_{\leqslant j<k\leqslant N}\ker L_{jk})$ be fixed. Then $|L_{jk}\omega_0|>0$ holds. In particular, there exist $\epsilon_0>0$ and $\delta>0$ such that $|L_{jk}\omega|\geqslant\epsilon_0$ for any $\omega\in\mathbb{S}^{d(N-1)-1}\cap B_{d(N-1)}(\omega_0,\delta)$. We then define a conical region
\begin{equation}
\mathscr{C}_{\omega_0,\delta}=\{\xi\in\mathbb{R}^{d(N-1)}\setminus\{0\}\bigm|\xi/|\xi|\in\mathbb{S}^{d(N-1)-1}\cap B_{d(N-1)}(\omega_0,\delta)\}.
\end{equation}
By this definition, if $\xi\in\mathscr{C}_{\omega_0,\delta}$, then $|\xi_{jk}|=|L_{jk}\xi|\geqslant\epsilon_0|\xi|$ holds for any $1\leqslant j<k\leqslant N$. In the following Proposition \ref{prop3} and proof of Proposition \ref{prop4}, we use $\omega_0\in\mathbb{S}^{d(N-1)-1}\setminus(\cup_{\leqslant j<k\leqslant N}\ker L_{jk})$, $\epsilon_0>0$, and $\delta>0$ chosen above.

\begin{prop}\label{prop3}
Let $\phi\in\mathscr{S}(\mathbb{R}^{d(N-1)})$ be such that $\mathscr{F}_{d(N-1)}\phi\in C_0^\infty(\mathbb{R}^{d(N-1)})$ with $\supp\mathscr{F}_{d(N-1)}\phi\subset\{\xi\in\mathscr{C}_{\omega_0,\delta}\bigm||\xi|\geqslant\epsilon\}$ for $0<\epsilon\leqslant\epsilon_0$, where $\mathscr{F}_{d(N-1)}$ denotes the Fourier transform on $L^2(\mathbb{R}^{d(N-1)})$. Then
\begin{gather}
\|F(\sqrt{\mu_{jk}}|r_k-r_j|\leqslant\epsilon^2t^{1-2\lambda}/(1-2\lambda))e^{-{\rm i}t^{1-2\lambda}D^2/(2(1-2\lambda))}\phi\|\nonumber\\
\lesssim_\nu t^{(d/2-\nu)(1-2\lambda)}\|\langle r_k-r_j\rangle^\nu\phi\|\label{prop3_1}
\end{gather}
holds for any $\nu\in\mathbb{N}$, $t\geqslant T$ and $1\leqslant j<k\leqslant N$, where $\mu_{jk}=m_jm_k/(m_j+m_k)$ and $\lesssim_\nu$ means that the constant depends on $\nu$.
\end{prop}

\begin{proof}[Proof of Proposition \ref{prop3}]
We have
\begin{gather}
F(\sqrt{\mu_{jk}}|r_k-r_j|\leqslant\epsilon^2t^{1-2\lambda}/(1-2\lambda))e^{-{\rm i}t^{1-2\lambda}D^2/(2(1-2\lambda))}\phi\nonumber\\
=e^{-{\rm i}t^{1-2\lambda}\sum_{j=2}^{N-1}D_j^2/(2(1-2\lambda))}F(|x_1|\leqslant\epsilon^2t^{1-2\lambda}/(1-2\lambda))e^{-{\rm i}t^{1-2\lambda}D_1^2/(2(1-2\lambda))}\phi,
\end{gather}
where we reconstruct the Jacobi coordinates by relabeling them so that the relative coordinate $y_1 = r_k - r_j$ appears as the first Jacobi component for $(j,k)\neq(1,2)$. Because $|\xi_1|=|L_{jlk}\xi|\geqslant\epsilon_0|\xi|\geqslant\epsilon^2$ holds on $\supp\mathscr{F}_{d(N-1)}\phi$, it suffices to prove that
\begin{gather}
\|F(|x_1|\leqslant\epsilon^2t^{1-2\lambda}/(1-2\lambda))e^{-{\rm i}t^{1-2\lambda}D_1^2/(2(1-2\lambda))}\psi\|_{L^2(\mathbb{R}^d)}\nonumber\\
\lesssim_\nu t^{(d/2-\nu)(1-2\lambda)}\|\langle x_1\rangle^\nu\psi\|_{L^2(\mathbb{R}^d)}\label{prop3_2}
\end{gather}
for $\mathscr{F}_d\psi\in C_0^\infty(\mathbb{R}^d)$ with $\supp\mathscr{F}_d\psi\subset\{\xi\in\mathbb{R}^d\bigm||\xi|\geqslant\epsilon^2\}$ where $\mathscr{F}_d$ denotes the Fourier transform on $L^2(\mathbb{R}^d)$.
We write
\begin{gather}
F(|x_1|\leqslant\epsilon^2t^{1-2\lambda}/(1-2\lambda))e^{-{\rm i}t^{1-2\lambda}D_1^2/(2(1-2\lambda))}\psi\nonumber\\
=\int_{\mathbb{R}^d}e^{{\rm i}(x_1\cdot\xi-t^{1-2\lambda}|\xi|^2/(2(1-2\lambda)))}F(|x_1|\leqslant\epsilon^2t^{1-2\lambda}/(1-2\lambda))\mathscr{F}_d\psi(\xi){\rm d}\xi/(2\pi)^{d/2}.
\end{gather}
When $|x_1|\leqslant\epsilon^2t^{1-2\lambda}/(1-2\lambda)$, we have
\begin{equation}
|x_1-t^{1-2\lambda}\xi/(1-2\lambda)|\geqslant\epsilon^2t^{1-2\lambda}/(1-2\lambda)
\end{equation}
on $\supp\mathscr{F}_d\psi$. By the relation
\begin{gather}
e^{{\rm i}(x_1\cdot\xi-t^{1-2\lambda}\xi^2/(2(1-2\lambda)))}\nonumber\\
=(x_1-t^{1-2\lambda}\xi/(1-2\lambda))\cdot(-{\rm i}\nabla_\xi)e^{{\rm i}(x_1\cdot\xi-t^{1-2\lambda}\xi^2/(2(1-2\lambda)))}/|x_1-t^{1-2\lambda}\xi/(1-2\lambda)|^2
\end{gather}
and integrating by parts, we have
\begin{gather}
|F(|x_1|\leqslant\epsilon^2t^{1-2\lambda}/(1-2\lambda))e^{-{\rm i}t^{1-2\lambda}D_1^2/(2(1-2\lambda))}\psi(x_1)|\nonumber\\
\lesssim t^{-\nu(1-2\lambda)}F(|x_1|\leqslant\epsilon^2t^{1-2\lambda}/(1-2\lambda))\|\langle x\rangle^\nu\psi\|_{L^2(\mathbb{R}_x^d)}.\label{prop3_3}
\end{gather}
Integrating \eqref{prop3_3} with respect to $x_1$ yields \eqref{prop3_2}.
\end{proof}

\begin{prop}\label{prop4}
The wave operators \eqref{wave_operator1} exist.
\end{prop}

\begin{proof}[Proof of Proposition \ref{prop4}]
We prove the existence of \eqref{wave_operator2} for $t\rightarrow\infty$. The derivative at $t$ of $U_{\lambda}(t,T)^*e^{-{\rm i}t^{1-2\lambda}D^2/(2(1-2\lambda))}$ is
\begin{gather}
({\rm d}/{\rm d}t)U_{\lambda}(t,T)^*e^{-{\rm i}t^{1-2\lambda}D^2/(2(1-2\lambda))}\nonumber\\
={\rm i}\sum_{1\leqslant j<k\leqslant N}U_{\lambda}(t,T)^*V_{jk}(t^\lambda(r_k-r_j))e^{-{\rm i}t^{1-2\lambda}D^2/(2(1-2\lambda))}.
\end{gather}
By Assumption \ref{ass},
\begin{equation}
\|V_{jk}^{\rm bdd}(t^\lambda(r_k-r_j))F(\sqrt{\mu_{jk}}|r_k-r_j|>\epsilon^2t^{1-2\lambda}/(1-2\lambda))\|\lesssim t^{-\rho(1-\lambda)}\label{prop4_1}
\end{equation}
holds for any $0<\epsilon\leqslant\epsilon_0$. Whereas we have
\begin{equation}
V_{jk}^{\rm sing}(t^\lambda(r_k-r_j))F(\sqrt{\mu_{jk}}|r_k-r_j|>\epsilon^2t^{1-2\lambda}/(1-2\lambda))\langle D/t^\lambda\rangle^{-2}=0\label{prop4_2}
\end{equation}
for $t\gg1$ because $V_{jk}^{\rm sing}$ is compactly supported. By Proposition \ref{prop3}, we clearly have
\begin{gather}
\|V_{jk}^{\rm bdd}(t^\lambda(r_k-r_j))F(\sqrt{\mu_{jk}}|r_k-r_j|\leqslant\epsilon^2t^{1-2\lambda}/(1-2\lambda))e^{-{\rm i}t^{1-2\lambda}D^2/(2(1-2\lambda))}\phi\|\nonumber\\
\lesssim t^{(d/2-\nu)(1-2\lambda)}\|\langle r_k-r_j\rangle^\nu\phi\|\label{prop4_3}
\end{gather}
because $V_{jk}^{\rm bdd}$ is bounded. Let $\chi\in C^\infty(\mathbb{R}^d)$ such that
\begin{equation}
\chi(x)=
\begin{cases}
\ 1&\mbox{if}\quad|x|\leqslant2\\
\ 0&\mbox{if}\quad|x|\geqslant3.
\end{cases}
\end{equation}
Reconstructing the Jacobi coordinates for $(j,k)\not=(1,2)$, we have
\begin{gather}
\|V_{jk}^{\rm sing}(t^\lambda(r_k-r_j))F(\sqrt{\mu_{jk}}|r_k-r_j|\leqslant\epsilon^2t^{1-2\lambda}/(1-2\lambda))e^{-{\rm i}t^{1-2\lambda}D^2/(2(1-2\lambda))}\phi\|\nonumber\\
\leqslant\|V_{jk}^{\rm sing}(t^\lambda x_1/\sqrt{\mu_{jk}})\chi((1-2\lambda)x_1/(\epsilon^2t^{1-2\lambda}))e^{-{\rm i}t^{1-2\lambda}D_1^2/(2(1-2\lambda))}\phi\|.
\end{gather}
Noting that $\|V_{jk}^{\rm sing}(t^\lambda x_1/\sqrt{\mu_{jk}})\langle\sqrt{\mu_{jk}}D_1/t^\lambda\rangle^{-2}\|=\|V_{jk}^{\rm sing}(x_1)\langle D_1\rangle^{-2}\|$ and that
\begin{gather}
{\rm i}[D_1^2,\chi((1-2\lambda)x_1/(\epsilon^2t^{1-2\lambda}))]=((1-2\lambda)/\epsilon^2t^{1-2\lambda})(\nabla\chi)((1-2\lambda)x_1/(\epsilon^2t^{1-2\lambda}))\cdot D_1\nonumber\\
-{\rm i}((1-2\lambda)/\epsilon^2t^{1-2\lambda})^2(\Delta\chi)((1-2\lambda)x_1/(\epsilon^2t^{1-2\lambda})),
\end{gather}
we have
\begin{gather}
\|\langle\sqrt{\mu_{jk}}D_1/t^\lambda\rangle^2\chi((1-2\lambda)x_1/(\epsilon^2t^{1-2\lambda}))e^{-{\rm i}t^{1-2\lambda}D_1^2/(2(1-2\lambda))}\phi\|\nonumber\\
\lesssim\|\chi((1-2\lambda)x_1/(\epsilon^2t^{1-2\lambda}))e^{-{\rm i}t^{1-2\lambda}D_1^2/(2(1-2\lambda))}\phi\|\nonumber\\
+t^{-1}\|(\nabla\chi)((1-2\lambda)x_1/(\epsilon^2t^{1-2\lambda}))\cdot e^{-{\rm i}t^{1-2\lambda}D_1^2/(2(1-2\lambda))}D_1\phi\|+t^{-2+2\lambda}\|\phi\|.\label{prop4_4}
\end{gather}
Choosing $\epsilon>0$ sufficiently small, we can use Proposition \ref{prop3} again for the first and second terms on the right side of \eqref{prop4_4}. By \eqref{prop4_1}, \eqref{prop4_2}, \eqref{prop4_3}, and \eqref{prop4_4}. We therefore have 
\begin{equation}
\|({\rm d}/{\rm d}t)U_{\lambda}(t,T)^*e^{-{\rm i}t^{1-2\lambda}D^2/(2(1-2\lambda))}\phi\|\lesssim t^{-\rho(1-\lambda)}+t^{(d/2-\nu)(1-2\lambda)}+t^{-2+2\lambda}
\end{equation}
for $t\gg1$. We note that $\rho(1-\lambda)>1$, $2-2\lambda>1$, and $(\nu-d/2)(1-2\lambda)>1$ for $\nu\gg1$. By the Cook-Kuroda method (\cite[Theorem XI.4]{ReSi2}), this completes the proof. 
\end{proof}

\section{Inverse Scattering}\label{inverse}
To prove Theorem \ref{thm1}, it is convenient to choose different Jacobi coordinates from those of \eqref{jacobi1}. According to \cite{EnWe, VaWe, We}, we focus on the pair of particles $1$ and $2$ and prove that $V_{12,1}=V_{12,2}$ because the proofs for the other pairs for $j$ and $k$ can be demonstrated in the same analogy. Let $y_{12}=r_2-r_1$ be the relative position between $r_1$ and $r_2$, and
\begin{equation}
y_j=r_j-(m_1r_1+m_2r_2)/(m_1+m_2)
\end{equation}
for $3\leqslant j\leqslant N$ be the relative position between $r_j$ and the center of mass for $r_1$ and $r_2$. Their conjugate momentums are given by
\begin{equation}
-{\rm i}\nabla_{y_{12}}=\mu_{12}(-{\rm i}\nabla_{r_2}/m_2-(-{\rm i}\nabla_{r_1})/m_1)\label{d_12}
\end{equation}
and
\begin{equation}
-{\rm i}\nabla_{y_j}=\mu_j(-{\rm i}\nabla_{r_j}/m_j-(-{\rm i}\nabla_{r_1}-{\rm i}\nabla_{r_2})/(m_1+m_2))
\end{equation}
for $3\leqslant j\leqslant N$ where
\begin{equation}
\mu_{12}=m_1m_2/(m_1+m_2)
\end{equation}
and
\begin{equation}
\mu_j=m_j(m_1+m_2)/(m_j+m_1+m_2)
\end{equation}
for $3\leqslant j\leqslant N$ are the reduced masses respectively. Because of the relations
\begin{gather}
\sum_{j=1}^Nm_jr_j^2=Mr_{\rm cm}^2+\mu_{12}y_{12}^2+\sum_{j=3}^{N}\mu_jy_j^2,\\
\sum_{j=1}^N\Delta_{r_j}/m_j=\Delta_{\rm cm}/M+\Delta_{y_{12}}/\mu_{12}+\sum_{j=3}^{N}\Delta_{y_j}/\mu_j,
\end{gather}
we rewrite the free Hamiltonian $H_0(t)$ in \eqref{free3} using $D=(D_{12},D_3,\ldots,D_N)$ and $x=(x_{12},x_3\ldots,x_N)$ with $x_{12}=\sqrt{\mu_{12}}y_{12}$ and $x_j=\sqrt{\mu_j}y_j$ for $3\leqslant j\leqslant N$. The relative momentum $D_{jk}$ is defined such that $D_{1k}=D_{2k}=D_k-D_{12}$ for $3\leqslant k\leqslant N$ and $D_{jk}=D_k-D_j$ for $3\leqslant j<k\leqslant N$. These coordinates are different from those in \cite{EnWe, VaWe, We} because of the scalings of $x_{12}$ and $x_j$ for $3\leqslant j\leqslant N$. By virtue of these scalings, we omit the effects of reduced masses in our notations. Even in these coordinates, the wave operators \eqref{wave_operator1} exist by the same proofs in Section \ref{existence}. In particular, the strong limits
\begin{equation}
W_\lambda^\pm=\slim_{t\rightarrow\pm\infty}U(t,0)^*U_{0,\lambda}(t)
\end{equation}
also exist, and we define
\begin{equation}
S_\lambda(V)=(W_\lambda^+)^*W_\lambda^-.
\end{equation}
Because $W^\pm=W_\lambda^\pm U_{0,\lambda}(s_\pm)^*U_0(s_\pm,0)$ holds for $s_+\geqslant T$ and $s_-\leqslant -T$, we have
\begin{equation}
S(V)=U_0(s_+,0)^*U_{0,\lambda}(s_+)S_\lambda(V)U_{0,\lambda}(s_-)^*U_0(s_-,0).\label{unitary_equiv}
\end{equation}
This implies that $S(V_1)=S(V_2)$ is equivalent to $S_\lambda(V_1)=S_\lambda(V_2)$.

$H_0(t)\equiv H_0=D^2/2+\omega^2 x^2/2$ is the time-independent harmonic oscillator for $|t|<T$, and its time-evolution $e^{-{\rm i}tH_0}$ is governed by the Mehler formula in \cite[Section 2.2]{BoCaHaMi} and \cite[Theorem 5.29]{LoHiBe}, which is
\begin{equation}
e^{-{\rm i}tH_0}=\mathscr{M}_{d(N-1)}(\tan\omega t/\omega)\mathscr{D}_{d(N-1)}(\sin\omega t/\omega)\mathscr{F}_{d(N-1)}\mathscr{M}(\tan\omega t/\omega)\label{mehler1}
\end{equation}
for $\omega t\notin\pi\mathbb{Z}$, where $\mathscr{M}_{d(N-1)}$ and $\mathscr{D}_{d(N-1)}$ are multiplication and dilation given by
\begin{gather}
\mathscr{M}_{d(N-1)}(t)\phi(x)=e^{{\rm i}x^2/(2t)}\phi(x),\\
\mathscr{D}_{d(N-1)}(t)\phi(x)=({\rm i}t)^{-d(N-1)/2}\phi(x/t).
\end{gather}
Moreover, 
\begin{equation}
e^{-{\rm i}tH_0}={\rm i}^{d(N-1)/2}\mathscr{M}_{d(N-1)}(-\cot\omega t/\omega)\mathscr{D}_{d(N-1)}(\cos\omega t)e^{-{\rm i}\tan\omega tD^2/(2\omega)}\label{mehler2}
\end{equation}
for $\omega t\notin(\pi/2)\mathbb{Z}$.

For $v\in\mathbb{R}^d$, the normalization is $\hat{v}=v/|v|$. Let $\hat{v}, e_3,\ldots,e_{N}\in\mathbb{R}^d$ be unit vectors pointing in mutually different directions. We define $v_{jk}=v_k-v_j$ for $1\leqslant j<k\leqslant N$ where $v_1=-v/2$, $v_2=v/2$, and $v_j=|v|^2e_j$ for $3\leqslant j\leqslant N$. Let $\Phi_0\in\mathscr{S}(\mathbb{R}^{d(N-1)})$ be $\Phi_0=\phi_1\otimes\phi_2$ such that $\mathscr{F}_d\phi_1\in C_0^\infty(\mathbb{R}^d)$ and $\mathscr{F}_{d(N-2)}\phi_2\in C_0^\infty(\mathbb{R}^{d(N-2)})$ with $\|\phi_2\|_{L^2(\mathbb{R}^{d(N-2)})}=1$. In addition, we assume that $\supp\mathscr{F}_d\phi_1\subset\{\xi_{12}\in\mathbb{R}^d\bigm||\xi_{12}|<\eta_{12}\}$ for some $\eta_{12}>0$ and $\supp\mathscr{F}_{d(N-2)}\phi_2\subset\{(\xi_3,\ldots,\xi_N)\in\mathbb{R}^{d(N-2)}\bigm||\xi_3|<1,\ldots,|\xi_N|<1\}$. We define $\Phi_v=\mathscr{T}_v\Phi_0$ for
\begin{equation}
\mathscr{T}_v=e^{{\rm i}v\cdot x_{12}}\prod_{j=3}^Ne^{{\rm i}v_j\cdot x_j}.
\end{equation}
Then there exists $f_{jk}\in C_0^\infty(\mathbb{R}^d)$ with $\supp f_{jk}\subset\{\xi\in\mathbb{R}^d\bigm||\xi|\leqslant\eta_{jk}\}$ where $\eta_{1k}=\eta_{2k}=1+\eta_{12}$ for $3\leqslant k\leqslant N$ and $\eta_{jk}=2$ for $3\leqslant j<k\leqslant N$ such that
\begin{equation}
\Phi_v=\mathscr{T}_vf_{jk}(D_{jk})\Phi_0=f_{jk}(D_{jk}-v_{jk})\Phi_v.\label{phi_v}
\end{equation}

To apply the Enss-Weder time-dependent method \cite{EnWe}, the following reconstruction formula is essential. By virtue of this formula, the uniqueness of the potential functions follows from the injective property of the Radon transform.

\begin{thm}\label{thm5}
Let $\Phi_v$ be defined as \eqref{phi_v} and $\Psi_v$ have the same properties for $\Psi_0=\psi_1\otimes\phi_2$. Then
\begin{equation}
\lim_{|v|\rightarrow\infty}|v|({\rm i}(S_\lambda(V)-1)\Phi_v,\Psi_v)=\int_{-\infty}^\infty(V_{12}(r_2-r_1+\hat{v}t)\phi_1,\psi_1)_{L^2(\mathbb{R}^d)}{\rm d}t
\end{equation}
holds for $V=\sum_{1\leqslant j<k\leqslant k}V_{jk}$, which satisfies Assumption \ref{ass}.
\end{thm}

It seems difficult to study the time evolution of $U_{0}(t,0)$ directly for all $t\in\mathbb{R}$. Instead, we analyze $e^{-{\rm i}tH_0}$ for $|t|<T$ and $e^{\mp{\rm i}|t|^{1-2\lambda}D^2/(2(1-2\lambda))}$ for $|t|\geqslant T$ according to \cite{Is7}. We define $U_{0,\lambda}(t)=e^{-{\rm i}tH_0}$ for $|t|<T$. In the proofs below, we can assume that
\begin{equation}
\pi/(2\omega)\leqslant T<\pi/\omega
\end{equation} 
without loss of generality, as in \cite{Is7}. 

\begin{lem}\label{lem7}
Let $\Phi_v$ be as in Theorem \ref{thm5}. Then
\begin{equation}
\int_{-\infty}^\infty\|V_{jk}^{\rm bdd}(r_k-r_j)U_{0,\lambda}(t)\Phi_v\|{\rm d}t=O(|v_{jk}|^{-1})
\end{equation}
holds as $|v|\rightarrow\infty$ for any $1\leqslant j<k\leqslant N$.
\end{lem}

\begin{proof}[Proof of Lemma \ref{lem7}]
We divide the integral such that
\begin{equation}
\int_{-\infty}^\infty=\int_{|t|<T}+\int_{|t|\geqslant T}
\end{equation}
and first consider the integral on $|t|<T$. By \eqref{mehler2} and \eqref{phi_v}, we have
\begin{gather}
\|V_{jk}^{\rm bdd}(r_k-r_j)e^{-{\rm i}tH_0}\Phi_v\|=\|V_{jk}^{\rm bdd}(\cos\omega t(r_k-r_j))e^{-{\rm i}\tan\omega tD^2/(2\omega)}f_{jk}(D_{jk}-v_{jk})\Phi_v\|\nonumber\\
=\|V_{jk}^{\rm bdd}(\cos\omega tx_{12}/\sqrt{\mu_{jk}})e^{-{\rm i}\tan\omega tD_{12}^2/(2\omega)}f_{jk}(D_{12}-v_{jk})\Phi_v\|\nonumber\\
=\|V_{jk}^{\rm bdd}((\cos\omega tx_{12}+\sin\omega tv_{jk}/\omega)/\sqrt{\mu_{jk}})e^{-{\rm i}\tan\omega tD_{12}^2/(2\omega)}f_{jk}(D_{12})e^{-{\rm i}v_{jk}\cdot x_{12}}\Phi_v\|,\label{lem7_1}
\end{gather}
where we have reconstructed the Jacobi coordinates from $y_{12}=r_k-r_j$ for $(j,k)\not=(1,2)$ and used the relation
\begin{equation}
e^{-{\rm i}v_{jk}\cdot x_{12}}e^{-{\rm i}\tan\omega tD_{12}^2/(2\omega)}e^{{\rm i}v_{jk}\cdot x_{12}}=e^{-{\rm i}\tan\omega t|v_{jk}|^2/(2\omega)}e^{-{\rm i}\tan\omega tD_{12}\cdot v_{jk}/\omega} e^{-{\rm i}\tan\omega tD_{12}^2/(2\omega)}.\label{lem7_2}
\end{equation}
From \cite[(39), (40), and (43) in the proof of Lemma 2.3]{Is7}, we have
\begin{equation}
\int_{|t|<T}\|V_{jk}^{\rm bdd}(r_k-r_j)e^{-{\rm i}tH_0}\Phi_v\|{\rm d}t=O(|v_{jk}|^{-1}).\label{lem7_3}
\end{equation}
We next consider the integral on $|t|\geqslant T$. In particular, we consider $t\geqslant T$. By \eqref{factorization+} and \eqref{phi_v}, we have
\begin{gather}
\|V_{jk}^{\rm bdd}(r_k-r_j)U_{0,\lambda}(t)\Phi_v\|=\|V_{jk}^{\rm bdd}(t^\lambda(r_k-r_j))e^{-{\rm i}t^{1-2\lambda}D^2/(2(1-2\lambda))}f_{jk}(D_{jk}-v_{jk})\Phi_v\|\nonumber\\
=\|V_{jk}^{\rm bdd}(t^\lambda x_{12}/\sqrt{\mu_{jk}})e^{-{\rm i}t^{1-2\lambda}D_{12}^2/(2(1-2\lambda))}f_{jk}(D_{12}-v_{jk})\Phi_v\|\nonumber\\
=\|V_{jk}^{\rm bdd}((t^\lambda x_{12}+t^{1-2\lambda}v_{jk}/(1-2\lambda))/\sqrt{\mu_{jk}})e^{-{\rm i}t^{1-2\lambda}D_{12}^2/(2(1-2\lambda))}f_{jk}(D_{12})e^{-{\rm i}v_{jk}\cdot x_{12}}\Phi_v\|,\label{lem7_4}
\end{gather}
where we have reconstructed the Jacobi coordinates from $y_{12}=r_k-r_j$ for $(j,k)\not=(1,2)$ and used the relation
\begin{gather}
e^{-{\rm i}v_{jk}\cdot x_{12}}e^{-{\rm i}t^{1-2\lambda}D_{12}^2/(2(1-2\lambda))}e^{{\rm i}v_{jk}\cdot x_{12}}\nonumber\\
=e^{-{\rm i}t^{1-2\lambda}|v_{jk}|^2/(2(1-2\lambda))}e^{-{\rm i}t^{1-2\lambda}D_{12}\cdot v_{jk}/(1-2\lambda)} e^{-{\rm i}t^{1-2\lambda}D_{12}^2/(2(1-2\lambda))}.\label{lem7_5}
\end{gather}
By \cite[(47) and (51) in the proof of Lemma 2.3]{Is7}, we have
\begin{equation}
\int_{|t|\geqslant T}\|V_{jk}^{\rm bdd}(r_k-r_j)U_{0,\lambda}(t)\Phi_v\|{\rm d}t=O(|v_{jk}|^{-1})\label{lem7_6}.
\end{equation}
\eqref{lem7_3} and \eqref{lem7_6} complete the proof.
\end{proof}

\begin{lem}\label{lem8}
Let $\Phi_v$ be as in Theorem \ref{thm5}. Then
\begin{equation}
\int_{-\infty}^\infty\|V_{jk}^{\rm sing}(r_k-r_j)U_{0,\lambda}(t)\Phi_v\|{\rm d}t=O(|v_{jk}|^{-1})
\end{equation}
holds as $|v|\rightarrow\infty$ for any $1\leqslant j<k\leqslant N$.
\end{lem}

\begin{proof}[Proof of Lemma \ref{lem8}]
We divide the integral such that
\begin{equation}
\int_{-\infty}^\infty=\int_{|t|\leqslant\pi/(4\omega)}+\int_{\pi/(4\omega)<|t|<T}+\int_{|t|\geqslant T}
\end{equation}
and first consider the integral on $|t|\leqslant\pi/(4\omega)$. As we similarly had for \eqref{lem7_1}, we have
\begin{gather}
\|V_{jk}^{\rm sing}(r_k-r_j)e^{-{\rm i}tH_0}\Phi_v\|=\|V_{jk}^{\rm sing}((\cos\omega tx_{12}+\sin\omega tv_{jk}/\omega)/\sqrt{\mu_{jk}})\langle D_{12}/\cos\omega t\rangle^{-2}\nonumber\\
\times e^{-{\rm i}\tan\omega tD_{12}^2/(2\omega)}f_{jk}(D_{12})\langle D_{12}/\cos\omega t\rangle^2e^{-{\rm i}v_{jk}\cdot x_{12}}\Phi_v\|.
\end{gather}
We take $g_{jk}\in C_0^\infty(\mathbb{R}^d)$ such that $f_{jk}=f_{jk}g_{fk}$. Noting that
\begin{gather}
\|\langle x_{12}\rangle^2\langle D_{12}/\cos\omega t\rangle^2g_{jk}(D_{12})e^{-{\rm i}v_{jk}\cdot x_{12}}\Phi_v\|\nonumber\\
\lesssim\|\langle D_{12}\rangle^2g_{jk}(D_{12})\|\|\langle x_{12}\rangle^2\Phi_0\|+\|[x_{12}^2,\langle D_{12}\rangle^2g_{jk}(D_{12})]e^{-{\rm i}v_{jk}\cdot x_{12}}\Phi_v\|\lesssim1,\label{lem8_1}
\end{gather}
we have
\begin{equation}
\int_{|t|\leqslant\pi/(4\omega)}\|V_{jk}^{\rm sing}(r_k-r_j)e^{-{\rm i}tH_0}\Phi_v\|=O(|v_{jk}|^{-1})\label{lem8_2}
\end{equation}
by \cite[(74) and (79) in the proof of Lemma 3.2]{Is7}. We next consider the integral on $\pi/(4\omega)<|t|<T$. We define the $d$-dimensional harmonic oscillator
\begin{equation}
H_{0,12}=D_{12}^2/2+w^2x_{12}^2/2
\end{equation}
and note that the Mehler formula
\begin{equation}
e^{-{\rm i}H_{0,12}}={\rm i}^{d/2}\mathscr{M}_d(-\cot\omega t/\omega)\mathscr{D}_d(\cos\omega t)e^{-{\rm i}\tan\omega tD_{12}^2/(2\omega)}\label{mehler3}
\end{equation}
also holds. As we similarly had for \eqref{lem7_1}, we have
\begin{gather}
\|V_{jk}^{\rm sing}(r_k-r_j)e^{-{\rm i}tH_0}\Phi_v\|\nonumber\\
=\|V_{jk}^{\rm sing}((\cos\omega tx_{12}+\sin\omega tv_{jk}/\omega)/\sqrt{\mu_{jk}})e^{-{\rm i}\tan\omega tD_{12}^2/(2\omega)}e^{-{\rm i}v_{jk}\cdot x_{12}}\Phi_v\|\nonumber\\
=\|V^{\rm sing}((x_{12}+\sin\omega tv_{jk}/\omega)/\sqrt{\mu_{jk}})e^{-{\rm i}H_{0,12}}e^{-{\rm i}v_{jk}\cdot x_{12}}\Phi_v\|\label{lem8_3}
\end{gather}
using \eqref{mehler3}. We therefore have
\begin{equation}
\int_{\pi/(4\omega)<|t|<T}\|V_{jk}^{\rm sing}(r_k-r_j)e^{-{\rm i}tH_0}\Phi_v\|=O(|v_{jk}|^{-1})\label{lem8_4}
\end{equation}
by \cite[(88) and (93) in the proof of Lemma 3.2]{Is7}. Finally, we consider the integral on $|t|\geqslant T$, in particular $t\geqslant T$. As we similarly had for \eqref{lem7_4}, we have
\begin{gather}
\|V_{jk}^{\rm sing}(r_k-r_j)U_{0,\lambda}(t)\Phi_v\|=\|V_{jk}^{\rm sing}((t^\lambda x_{12}+t^{1-2\lambda}v_{jk}/(1-2\lambda))/\sqrt{\mu_{jk}})\langle D_{12}/t^\lambda\rangle^{-2}\nonumber\\
\times e^{-{\rm i}t^{1-2\lambda}D_{12}^2/(2(1-2\lambda))}f_{jk}(D_{12})\langle D_{12}/t^\lambda\rangle^2e^{-{\rm i}v_{jk}\cdot x_{12}}\Phi_v\|.
\end{gather}
As in \eqref{lem8_1}, we have
\begin{equation}
\|\langle x_{12}\rangle^2\langle D_{12}/t^\lambda\rangle^2g_{jk}(D_{12})e^{-{\rm i}v_{jk}\cdot x_{12}}\Phi_v\|\lesssim1
\end{equation}
for $g_{jk}\in C_0^\infty(\mathbb{R}^d)$ such that $f_{jk}=f_{jk}g_{jk}$. We inductively have
\begin{equation}
\|\langle x_{12}\rangle^2\langle D_{12}/t^\lambda\rangle^Ng_{jk}(D_{12})e^{-{\rm i}v_{jk}\cdot x_{12}}\Phi_v\|\lesssim1
\end{equation}
for any $N\in\mathbb{N}$. Therefore, by \cite[(97) and (99) in the proof of Lemma 3.2]{Is7}, we have
\begin{equation}
\int_{|t|\geqslant T}\|V_{jk}^{\rm sing}(r_k-r_j)U_{0,\lambda}(t)\Phi_v\|=O(|v_{jk}|^{-1}).\label{lem8_5}
\end{equation}
\eqref{lem8_2}, \eqref{lem8_4}, and \eqref{lem8_5} imply that Lemma \ref{lem8} holds.
\end{proof}

By virtue of Lemmas \ref{lem7} and \ref{lem8}, we immediately have the following Lemma as in \cite[Corollary 2.3]{EnWe} (see also \cite{AdFuIs, AdKaKaTo, AdMa, AdTs1, AdTs2, Is1, Is3, Is4, Is7, Is8, Ni1, Ni2, Ni3, VaWe, We}). We therefore omit its proof.
\begin{lem}\label{lem9}
Let $\Phi_v$ be as in Theorem \ref{thm5}. Then
\begin{equation}
\sup_{t\in\mathbb{R}}\|(U(t,0)W_\lambda^--U_{0,\lambda}(t))\Phi_v\|=O(|v|^{-1})
\end{equation}
holds as $|v|\rightarrow\infty$.
\end{lem}

\begin{proof}[Proof of Theorem \ref{thm5}]
We have
\begin{gather}
{\rm i}(S_\lambda(V)-1)={\rm i}(W_\lambda^+-W_\lambda^-)^*W_\lambda^-={\rm i}\int_{-\infty}^\infty(({\rm d}/{\rm d}t)U(t,0)^*U_{0,\lambda}(t))^*W_\lambda^-{\rm d}t\nonumber\\
=\sum_{1\leqslant j<k\leqslant N}\int_{-\infty}^\infty U_{0,\lambda}(t)^*V_{jk}(r_k-r_j)U(t,0)W_\lambda^-{\rm d}t
\end{gather}
and
\begin{equation}
|v|(({\rm i}S(V)-1)\Phi_v,\Psi_v)=|v|\int_{-\infty}^\infty(V_{12}(r_2-r_1)U_{0,\lambda}(t)\Phi_v,U_{0,\lambda}(t)\Psi_v){\rm d}t+R(v),
\end{equation}
where
\begin{gather}
R(v)=|v|\sum_{(j,k)\not=(1,2)}\int_{-\infty}^\infty(V_{jk}(r_k-r_j)U_{0,\lambda}(t)\Phi_v,U_{0,\lambda}(t)\Psi_v){\rm d}t\nonumber\\
+|v|\sum_{1\leqslant j<k\leqslant N}\int_{-\infty}^\infty((U(t,0)W_\lambda^--U_{0,\lambda}(t))\Phi_v,V_{jk}(r_k-r_j)U_{0,\lambda}(t)){\rm d}t.
\end{gather}
Noting that
\begin{equation}
\int_{-\infty}^\infty\|V_{jk}(r_k-r_j)U_{0,\lambda}(t)\Phi_v\|{\rm d}t=O(|v|^{-2})
\end{equation}
holds as $|v|\rightarrow\infty$ for $(j,k)\not=(1,2)$, we have $R(v)=O(|v|^{-1})$ by Lemmas \ref{lem7}, \ref{lem8}, and \ref{lem9}. We now prove
\begin{equation}
|v|\int_{-\infty}^\infty(V_{12}(r_2-r_1)U_{0,\lambda}(t)\Phi_v,U_{0,\lambda}(t)\Psi_v){\rm d}t\rightarrow\int_{-\infty}^\infty(V_{12}(r_2-r_1+\hat{v}t)\phi_1,\psi_1)_{L^2(\mathbb{R}^d)}{\rm d}t\label{thm5_1}
\end{equation}
as $|v|\rightarrow\infty$. We first focus on the term $V_{12}^{\rm bdd}$. We divide the integral such that
\begin{equation}
\int_{-\infty}^\infty=\int_{|t|<\pi/(2\omega)}+\int_{\pi/(2\omega)\leqslant|t|<T}+\int_{|t|\geqslant T}
\end{equation}
and consider the integrals on $|t|<\pi/(2\omega)$ and $\pi/(2\omega)\leqslant|t|<T$. Recalling that $\phi_2$ is normalized, we have
\begin{gather}
(V_{12}^{\rm bdd}(r_2-r_1)e^{-{\rm i}tH_0}\Phi_v,e^{-{\rm i}tH_0}\Psi_v)\nonumber\\
=(V_{12}^{\rm bdd}(\cos\omega tx_{12}/\sqrt{\mu_{12}})e^{-{\rm i}\tan\omega tD_{12}^2/(2\omega)}\Phi_v,e^{-{\rm i}\tan\omega tD_{12}^2/(2\omega)}\Psi_v)\nonumber\\
=(V_{12}^{\rm bdd}(x_{12}/\sqrt{\mu_{12}})e^{-{\rm i}H_{0,12}}e^{{\rm i}v\cdot x_{12}}\phi_1,e^{-{\rm i}H_{0,12}}e^{{\rm i}v\cdot x_{12}}\psi_1)_{L^2(\mathbb{R}^d)}\label{thm5_2}
\end{gather}
for $|t|<T$ by \eqref{mehler2} and \eqref{mehler3}. Therefore
\begin{gather}
|v|\int_{|t|<\pi/(2\omega)}(V_{12}^{\rm bdd}(r_2-r_1)e^{-{\rm i}tH_0}\Phi_v,e^{-{\rm i}tH_0}\Psi_v){\rm d}t\nonumber\\
\rightarrow\int_{-\infty}^\infty(V_{12}^{\rm bdd}(x_{12}/\sqrt{\mu_{12}}+\hat{v}t)\phi_1,\psi_1)_{L^2(\mathbb{R}^d)}{\rm d}t\label{thm5_3}
\end{gather}
and
\begin{equation}
|v|\int_{\pi/(2\omega)\leqslant|t|<T}(V_{12}^{\rm bdd}(r_2-r_1)e^{-{\rm i}tH_0}\Phi_v,e^{-{\rm i}tH_0}\Psi_v){\rm d}t\rightarrow0\label{thm5_4}
\end{equation}
hold as $|v|\rightarrow\infty$ by \cite[(64) and (65) in the proof of Theorem 2.1]{Is7}. Let $U_{0,12}(t,s)$ be the propagator for the $d$-dimensional time-dependent harmonic oscillator
\begin{equation}
H_{0,12}(t)=D_{12}^2/2+k(t)x_{12}^2/2
\end{equation}
and define
\begin{equation}
U_{0,12,\lambda}(t)=e^{{\rm i}\lambda x_{12}^2/(2t)}e^{-{\rm i}\lambda\log tA_{12}}e^{-{\rm i}t^{1-2\lambda}D_{12}^2/(2(1-2\lambda))}\label{factorization_d+}
\end{equation}
if $t\geqslant T$ and
\begin{equation}
U_{0,12,\lambda}(t)=e^{{\rm i}\lambda x_{12}^2/(2t)}e^{-{\rm i}\lambda\log(-t)A_{12}}e^{{\rm i}(-t)^{1-2\lambda}D_{12}^2/(2(1-2\lambda))}\label{factorization_d-}
\end{equation}
if $t\leqslant-T$, where $A_{12}=(D_{12}\cdot x_{12}+x_{12}\cdot D_{12})/2$. By \eqref{factorization+}, \eqref{factorization-}, \eqref{factorization_d+}, and \eqref{factorization_d-}, we have
\begin{gather}
(V_{12}^{\rm bdd}(r_2-r_1)U_{0,\lambda}(t)\Phi_v,U_{0,\lambda}(t)\Psi_v)\nonumber\\
=(V_{12}^{\rm bdd}(t^\lambda x_{12}/\sqrt{\mu_{12}})e^{-{\rm i}t^{1-2\lambda}D_{12}^2/(2(1-2\lambda))}\Phi_v,e^{-{\rm i}t^{1-2\lambda}D_{12}^2/(2(1-2\lambda))}\Psi_v)\nonumber\\
=(V_{12}^{\rm bdd}(x_{12}/\sqrt{\mu_{12}})U_{0,12,\lambda}(t)e^{{\rm i}v\cdot x_{12}}\phi_1,U_{0,12,\lambda}(t)e^{{\rm i}v\cdot x_{12}}\psi_1)_{L^2(\mathbb{R}^d)}\label{thm5_5}
\end{gather}
for $|t|\geqslant T$. By \cite[(66) in the proof of Theorem 2.1]{Is7},
\begin{equation}
|v|\int_{|t|\geqslant T}(V_{12}^{\rm bdd}(r_2-r_1)U_{0,\lambda}(t)\Phi_v,U_{0,\lambda}(t)\Psi_v){\rm d}t\rightarrow0\label{thm5_6}
\end{equation}
holds as $|v|\rightarrow\infty$. We next focus on the term $V_{12}^{\rm sing}$. We divide the integral such that
\begin{equation}
\int_{-\infty}^\infty=\int_{|t|<\pi/(4\omega)}+\int_{\pi/(4\omega)\leqslant|t|<T}+\int_{|t|\geqslant T}
\end{equation}
and consider the integrals on $|t|<\pi/(4\omega)$ and $\pi/(4\omega)\leqslant|t|<T$. Because \eqref{thm5_2} also holds with $V_{12}^{\rm bdd}$ replaced by $V_{12}^{\rm sing}$, we have
\begin{gather}
|v|\int_{|t|<\pi/(4\omega)}(V_{12}^{\rm sing}(r_2-r_1)e^{-{\rm i}tH_0}\Phi_v,e^{-{\rm i}tH_0}\Psi_v){\rm d}t\nonumber\\
\rightarrow\int_{-\infty}^\infty(V_{12}^{\rm sing}(x_{12}/\sqrt{\mu_{12}}+\hat{v}t)\phi_1,\psi_1)_{L^2(\mathbb{R}^d)}{\rm d}t\label{thm5_7}
\end{gather}
and
\begin{equation}
|v|\int_{\pi/(4\omega)\leqslant|t|<T}(V_{12}^{\rm sing}(r_2-r_1)e^{-{\rm i}tH_0}\Phi_v,e^{-{\rm i}tH_0}\Psi_v){\rm d}t\rightarrow0\label{thm5_8}
\end{equation}
as $|v|\rightarrow\infty$ by \cite[(110) and (111) in the proof of Theorem 3.1]{Is7}. Finally, we consider the integral on $|t|\geqslant T$. Because \eqref{thm5_5} also holds with $V_{12}^{\rm bdd}$ replaced by $V_{12}^{\rm sing}$, we have
\begin{equation}
|v|\int_{|t|\geqslant T}(V_{12}^{\rm bdd}(r_2-r_1)U_{0,\lambda}(t)\Phi_v,U_{0,\lambda}(t)\Psi_v){\rm d}t\rightarrow0\label{thm5_9}
\end{equation}
as $|v|\rightarrow\infty$ by \cite[(112) in the proof of Theorem 3.1]{Is7}. Combining \eqref{thm5_3}, \eqref{thm5_4}, \eqref{thm5_6}, \eqref{thm5_7}, \eqref{thm5_8}, and \eqref{thm5_9}, we have \eqref{thm5_1}.
\end{proof}

\begin{proof}[Proof of Theorem \ref{thm1}]
We assume that $S_\lambda(V_1)=S_\lambda(V_2)$ because \eqref{unitary_equiv} holds. By the same computation as that of \cite[(77) in the proof of Lemma 3.2]{Is7}, we have the condition commonly referred to as the Enss condition
\begin{equation}
\int_0^\infty\|V_{12}(r_2-r_1)\langle D_{12}\rangle^{-2}F(|r_2-r_1|\geqslant R)\|_{\mathscr{B}(L^2(\mathbb{R}^d))}{\rm d}R<\infty.
\end{equation}
By virtue of Theorem \ref{thm5} and the Plancherel formula for the Radon transform (\cite[Theorem 2.17 in Chapter 1]{He}), we have $V_{12,1}=V_{12,2}$ in the same way as for the proof of \cite[Theorem 1.1]{EnWe}.
\end{proof}

\bigskip
\noindent\textbf{Acknowledgments.} 
This work was supported by JSPS KAKENHI Grant Number JP21K03279. 

\end{document}